\documentclass[10pt, twocolumn]{IEEEtran}

\usepackage{amsmath}
\usepackage{amssymb}    
\usepackage{amsfonts}
\usepackage[english]{babel}
\usepackage{cite} 
\usepackage{multirow}
\usepackage{stfloats}    

\usepackage{algorithm}   
\usepackage{algorithmic} 

\usepackage{caption} 
\usepackage{graphicx}
\usepackage{epsfig}
\usepackage{subfigure} 
\usepackage{tablefootnote}
\usepackage{float}

\usepackage{url}

\usepackage{accents}
\makeatletter
\def\widebar{\accentset{{\cc@style\underline{\mskip10mu}}}}
\def\Widebar{\accentset{{\cc@style\underline{\mskip13mu}}}}
\makeatother

\newtheorem{theorem}{Theorem}
\newtheorem{remark}{Remark}

\newtheorem{definition}{Definition}
\newtheorem{proposition}{Proposition}

\begin{document}
\captionsetup[figure]{name={Fig.},labelsep=period}  

\title{Broadcast Age of Information in CSMA/CA Based Wireless Networks}
\author{\authorblockN{Mei Wang 
                                           and Yunquan Dong  \\ 
\authorblockA{ \normalsize
                                 School of Electronic and Information Engineering, \\
                                 Nanjing University of Information Science \& Technology, Nanjing, China \\
yunquandong@nuist.edu.cn
}
}
\thanks{ This work was supported in parts by the National Natural Science Foundation of China (NSFC) under Grant 61701247,  the Jiangsu Provincial Natural Science Research Project under Grant 17KJB510035,
and the Startup Foundation for Introducing Talent of NUIST under Grant 2243141701008.}
}
\maketitle
\thispagestyle{empty}

\begin{abstract}
We consider a wireless sensor network in which all the nodes wish to spread their updates over the network using CSMA/CA protocol.
    We investigate the age of information of the spreading process from a transmitter perspective, which is referred to as the \textit{broadcast age of information (BAoI)}.
 To be specific, BAoI is the age of the latest update successfully broadcasted to the one-hop neighbors of a node, and thus is suitable to measure the rapidity of the update spreading process.
    We establish an equivalent transmission model of the network by deriving the transmission probability and the collision probability of nodes.
 With this equivalent model, we then present the average BAoI of the network explicitly.
    Our results present the scaling laws of average BAoI with respect to node density and frame length, and are further illustrated through numerical results.
As is shown, the average BAoI is increasing with node density and is convex in frame length, i.e., would be large when frame length is very small or very large.

\end{abstract}


\section{Introduction}
With the fast development of real-time applications in the scenarios like Internet-of-Things and vehicular networks, age of information (AoI) was proposed as a promising timeless measure of information transmissions~\cite{Yates-2012-age}.
    Specifically, AoI is defined as the elapsed time since the generation of the latest received update, i.e., the age of the newest available update at the receiver.
It well known that AoI focuses on the freshness of received updates and would be large either when the throughput is very low or very high.
    Thus, AoI is widely considered as a better timeliness measure than throughput and delay and has been exhaustively studied in various queueing systems, e.g., $M/M/1, M/D/1$ and $D/M/1$~\cite{Yates-2012-age}, under several serving disciplines, e.g., the first-come-first-served (FCFS)~\cite{Yates-2012-age, Dong-2018-two-way} and the last-generate-first-served (LGFS)~\cite{Sun-2016-mlt-sver}.

Note that AoI specifies the freshness of information at the receiver and is very applicable to point-to-point communications.
    In many wireless networks, however, the nodes need to spread their updates over the whole network.
For example, when an accident occurs,  the car should spread this urgent message over the whole vehicular network as soon as possible so that other cars can make prompt and appropriate responses (e.g., changing lane for neighboring cars and slowing down for following cars).
    Under the AoI framework, we noted that the AoI of a car would be changed when it receives a new message from any of its neighbors, regardless when the message is generated and whether the message is more fresh than the previously received ones or not.
In this paper, therefore, we are motivated to investigate the one-hop \textit{broadcast age of information} (BAoI) from a transmitter perspective.
To be specific,
\begin{itemize}
  \item BAoI quantifies the freshness of each update broadcast of the nodes and characterizes their capability in promptly spreading updates over the network.
\end{itemize}

In particular, BAoI is defined as the elapsed time since the generation of the latest transmitted update, i.e., the age of the newest successful broadcast.
    Although following the similar definition of AoI, BAoI enables us to investigate the broadcasting behavior of each information source.
Note that in wireless networks, each node receives updates from several neighbors; however, AoI measures the freshness of received updates without differentiating the sources of received updates.
That is, the AoI of a node would be changed whenever it receives an update, regardless of its source and whether the update is more fresh or not.
    As a result, the physical meaning of AoI is not be very clear in this scenario.
On the contrary, BAoI characterizes the timeliness of updates received by all the one-hop neighbors of each information source, and thus is especially suitable to quantify the rapidity of status dissemination over wireless networks.

Since the carrier-sense multiple access/collision-avoidance (CSMA/CA) is the most widely used medium-access control scheme in wireless networks, CSMA/CA based broadcasting has been exhaustively studied, e.g.,~\cite{Bian-2000-JSAC-CSMA-th, Babak-2011-TVT-CSMA-th, Kar-2005-Infocom-CSMA-fair, Modiano-2000-TON-broadcst}.
    Firstly, the throughout performance of CSMA/CA based networks was analyzed in~\cite{Bian-2000-JSAC-CSMA-th}.
A comprehensive analysis and comparison of the broadcasting reliability of different routing protocols (e.g., flooding, $p$-flooding) in terms of network coverage was presented in~\cite{Babak-2011-TVT-CSMA-th}.
    The fairness among network nodes also was discussed in~\cite{Kar-2005-Infocom-CSMA-fair}.
Moreover, the authors proposed throughput-optimal broadcasting schemes over wireless networks in~\cite{Modiano-2000-TON-broadcst}.

In this paper, we consider a slotted CSMA/CA based wireless network with uniformly distributed nodes.
    In each frame, each node generates one status update and tries to broadcast it to its one-hop neighbors accord to the CSMA/CA protocol.
If the node receives an update from its neighbors in the frame, its generated update would be replaced by the received update.
    For this network, we shall analyze  the collision probability and the transmission probability of nodes, based on which we establish an equivalent service (transmission) model.
Afterwards, the average BAoI of the network shall be presented explicitly.
    We show that the transmission probability is decreasing with node density while the collision probability is increasing with node density.
As shown in our numerical results, the average BAoI is increasing with node density and is convex in frame length.
    That is, the average BAoI would be increased if frame length becomes very small or very large, and achieves its minimum at the middle.
Since one update is generated in each frame at each node, the frame length would be inversely proportional to traffic rate.
    Therefore, this result also shows the scaling between average BAoI and traffic rate.

The rest of the paper is organized as follows.
   In Section~\ref{sec:model}, we present our network model and the definition of BAoI.
In Section~\ref{Transmission Efficiency}, we establish an equivalent transmission model for the CSMA/CA based network by deriving the transmission probability and the collision probability of nodes.
    In section~\ref{sec:baoi}, we investigate the average BAoI of the network by modeling the broadcasting process as a $GI/G/1$ queue.
   Finally, numerical results are provided in Section~\ref{sec:nuresults} and our work is concluded in Section~\ref{sec:conclusion}.

\section{System Model}\label{sec:model}
We consider a wireless network where the nodes are distributed according to a Poisson point process with density $\rho$, as shown in Fig. \ref{fig:a wireless network}.
    The nodes observes a certain phenomenon individually and wants to share its observed status updates over the whole network.
We assume that time is slotted and $T_\text{F}$ slots make a frame.
    In each frame, each node would observe the phenomenon once and generate one status update in a uniformly distributed slot.
If no update from other nodes is received in the frame, its own update would be put into a update queue and broadcasted later.
    Otherwise, its update would be discarded and the received one will be saved and broadcasted.
Moreover, we assume that each update can be delivered using one slot.

We assume that each node can communicate with its neighbors if their distance is less than transmit range $r$.
    It is clear that the number of nodes in each area with radius $r$ is Poisson distributed with parameter $\lambda = \rho\pi$$r^2$.
We also assume that the transmit range $r$ is small and the channels between neighboring nodes are Gaussian channels.

We assume that all the nodes in the network shares the same channel and performs transmission according to the CSMA/CA protocol.
    In this protocol, a node should wait for a certain period (which is referred to as the \textit{back-off time}) before starting a transmission to avoid collisions.
The length of the back-off period is uniformly chosen between zero and a contention window $w$, where $w=2^m w_{\min}$ and $m=0, 1, \cdots$ is the number of unsuccessful back-offs.
    The passed back-off time is recorded with a timer, which would be decreased by one if the channel is idle and remain unchanged otherwise in a slot.
When the timer reaches zero, the node would perform a slot of transmission.
    If collision with the transmissions from other node(s) occurs, the node should go to another round of back-off and the number of back-offs should be updated accordingly.
We assume that the number of back-offs can be infinitely large.


Unlike the traditional AoI metric which measures the freshness of update from the receiver perspective, we propose a new metric termed as the \textit{broadcast age of information}.

\begin{definition}
\textit{Broadcast Age of Information (BAoI)} is defined as the difference between the current time and the generation time $U(m)$ of the latest successfully broadcasted packet.
    In block $m$, that is
   \begin{equation}
   \Delta_{\text{BC}}(m) = m - U(m).
   \end{equation}
\end{definition}

In this paper, we focus on studying the average one-hop BAoI of the network, which is defined as
\begin{equation}
     \bar{\Delta}_{\text{BC}}=\lim_{M\to\infty}\frac{1}{M} \sum^M_{m=1}\Delta_{\text{BC}}(m).
\end{equation}

\begin{figure}
  \centering
  \includegraphics[width=3.2in]{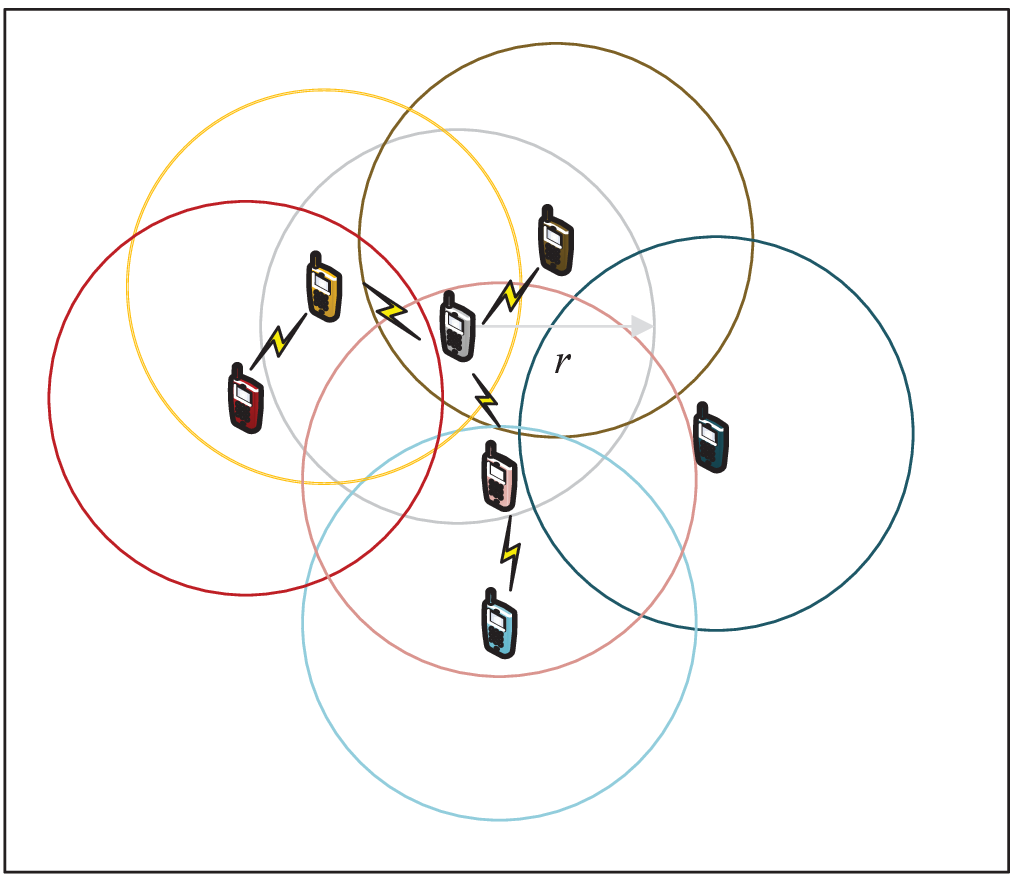}\\
  \caption{The wireless network model}\label{fig:a wireless network}
\end{figure}

\begin{figure}
  \centering
  \includegraphics[width=3.2in]{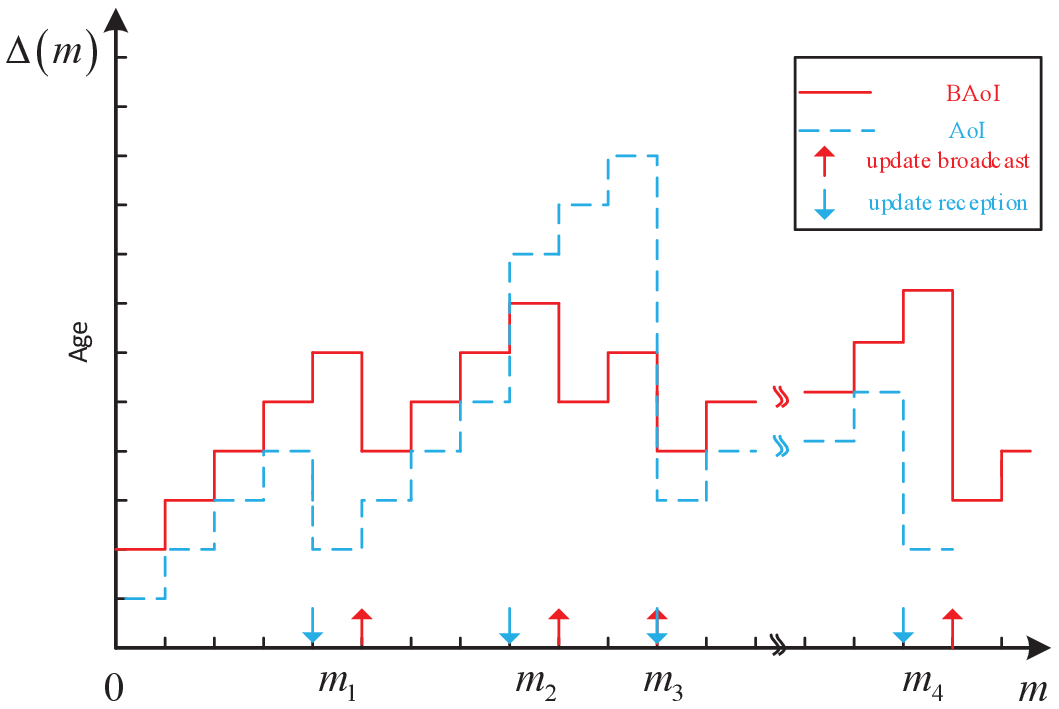}\\
  \caption{Sample paths of BAoI and AoI}\label{fig:BAoI AoI}
\end{figure}
Fig. \ref{fig:BAoI AoI} presents a sample path of the variations of BAoI and AoI.
    It is observed that BAoI increases linearly in time and is reset to a smaller value (the age of the newest successfully broadcast).
On the contrary, the physical meaning of AoI as shown by the dashed curve is not very clear in this scenario.
   To be specific, if a node receives a new update that is less fresh than the previously received one from some other neighbors, the AoI of the node will be reset to a larger value.

\section{Transmission Efficiency of the Network}\label{Transmission Efficiency}

In this section, we shall establish an equivalent transmission model for the CSMA/CA based network using \textit{transmission probability} $p_\text{tx}$ and \textit{collision probability} $p_\text{cl}$.
    Specifically, for each node, $p_\text{tx}$ is the probability that the backoff timer of the node decreases to zero and starts to transmit; $p_\text{cl}$ is the probability that the transmission of the node collides with that of some other node(s).

In particular, both $p_\text{tx}$ and $p_\text{cl}$ have been obtained when the maximum number of back-offs is finite using a two-dimensional Markov method in~\cite{Bian-2000-JSAC-CSMA-th}.
    By replacing the maximum number of back-offs with infinity, we have the following immediate results
\begin{align}\label{eq:p_tx}
  p_\text{tx} &= \frac{2 (1-2p_\text{cl})}{ w_\text{min}(1-p_\text{cl}) + (1-2p_\text{cl})}, \\
  \label{eq:p_cl}
  p_\text{cl} &= 1-\big(1-p_\text{tx}\big)^{N_\text{nb}-1},
\end{align}
where $N_\text{nb}$ is the number of neighbor nodes of the considered node.
    Moreover, it can be readily proved that $p_\text{cl}$ is guaranteed to be smaller than 0.5.
Or else, all the states of the two-dimensional Markov chain would be transient states~\cite{Bian-2000-JSAC-CSMA-th}.

Since the number of nodes $N$ in each area with radius $r$ is Poisson distributed, the probability that a node has $n$ neighbors would be
\begin{align}\label{the probability distribution of neighbor nodes}
   \Pr \{N_\text{nb} = n \} &= \Pr \{N-1=n | N\geqslant 1 \}  \nonumber  \\
   &= \frac{ (\rho \pi r^2)^{(1+n)} e^{-\rho \pi r^2} }{ (1+n)!(1-e^{-\rho \pi r^2}) }, ~~   n=0,1,2,\cdots
\end{align}

   Combining \eqref{eq:p_cl} and \eqref{the probability distribution of neighbor nodes}, the  average collision probability of each node can be written  as
\begin{align}\label{the average probability distribution of neighbor nodes}
   \bar{p}_\text{cl}&=\sum^\infty_{n=0} \Pr \{ N_\text{nb}=n \}\cdot p_\text{cl}|_{N_\text{nb}=n}  \nonumber \\
   &= 1- \frac{\rho\pi r^2 e^{-\rho\pi r^2}(p_\text{tx}^2-p_\text{tx}) + e^{-\rho\pi r^2 p_\text{tx}} -e^{-\rho\pi r^2} }{(1-e^{-\rho\pi r^2})(1-p_\text{tx})^2}.
\end{align}
We assume that in each block, each node would start to broadcast its updates with probability $p_\text{tx}$ and the probability that there would be collisions is $\bar{p}_\text{cl}$.
   \eqref{eq:p_tx} and \eqref{the average probability distribution of neighbor nodes} represent a nonlinear system in the two unknowns $p_\text{tx}$ and $\bar{p}_\text{cl}$, which can be solved efficiently using  numerical techniques.
We assume that in each block, each node would start to broadcast its updates with probability $p_\text{tx}$ and the probability that there would be collisions is $\bar{p}_\text{cl}$.

In each block, a node would start transmission and can deliver its updates successfully only if the node acquires the channel and no collision occurs.
    Therefore, the probability of successful broadcasting is given by
    \begin{equation}
        \mu = {(1-\bar{p}_\text{cl})}{p_\text{tx}},
    \end{equation}
which is referred to as the \textit{service rate} of an update.
We denote the time for a node to acquire the channel and starts successful broadcasting update $k$ as $S_k$, and thus know that $S_k$ is geometrically distributed with parameter $\mu$.
    That is, the probability that service time $S_k$ equals to $j$ is given by,
    \begin{equation} \label{eq:p_skj}
        \Pr\{S_k=j\} = (1-\mu)^{j-1}\mu, ~~ j=1,2,\cdots.
    \end{equation}

\section{BAoI of the Network} \label{sec:baoi}
\subsection{Queueing Model}
We assume that in each frame, each node generate an update and put it into a update queue at a random time which is uniformly distributed in (1,$T_\text{F}$).
    Therefore, the \textit{inter-arrival time} $X_k = m_{k+1}-m_{k}$ between the  arrival of update $k$ and update $k+1$ would be the sum of two uniformly distributed random variables in (1,$T_\text{F}$-1).
We then have
\begin{equation} \label{eq:p_xkj}
   \Pr \{X_k=j\}=\left\{
                    \begin{aligned}
                    &\frac{j}{T_\text{F}^2},        &  1\leq{j}\leq T_\text{F}, \\
                    &\frac{2T_\text{F}-j}{T_\text{F}^2},  &  T_\text{F}<j\leq 2T_\text{F}-1,
                    \end{aligned}
    \right.
\end{equation}
    $\mathbb{E}[X]$ = $T_\text{F}$, and $\mathbb{E}[X^2] = (7T_\text{F}^2-1)/6$.
We define an auxiliary function $h(z)$ as
\begin{equation} \label{df:hx}
  h(x) =\frac{x - 2x^{T_\text{F}+1}+x^{2T_\text{F}+1}}{(1-x)^2}.
\end{equation}
The probability generating function (PGF) of $X_i$ can then be given by the following proposition.

\begin{proposition} \label{prop:pgf_xi}
   The PGF of  inter-arrival time $X_i$ is given by
\begin{equation}
   G_X(z) =\frac{ h(z)}{T_\text{F}^2}
\end{equation}
where $h(x)$ is defined in \eqref{df:hx}.

\begin{proof}
    See Appendix \ref{proof:pgf_xi}.
\end{proof}
\end{proposition}

We have shown that the service time $S_k$ is geometrically distributed with average $1/\mu$ (c.f. \eqref{eq:p_skj}).
    Therefore, the broadcasting process of a node can be modeled by a $GI/G/1$ queue, where the arrival process is a general process defined by \eqref{eq:p_xkj} and the service process is defined by a geometric process.
According to \cite[Chap. ~5, Theorem 3]{book_queueing_2008}, the system time $T_k$, which is the sum of the waiting time and the service time of update $k$,  follows the geometric distribution with $\mathbb{E}(T)= \frac{1}{1-\nu}$, where
\begin{equation}\label{eq:alpha_equation}
  \nu = 1-\mu (1-\alpha)
\end{equation}
   and  $\alpha$ is the  solution of  $z=G_X[1-\mu(1-z)]$ within $(0,1)$.

%

\subsection{Broadcast Age of Information}
We denote the time between two consecutive broadcasted updates as the \textit{inter-departure time} $Y_k$.
    In particular, the PGF of $Y_k$ is given by the following proposition.

\begin{proposition} \label{prop:pgf yi}
    The PGF of the inter-departure time $Y_i$ is
    \begin{align}
    G_Y(z) = &\frac{z{\mu}h(\nu)}{T_{F}^2(1-z+z\mu)} + \left(1-\frac{h(\nu)}{T_\text{F}^2}\right) \nonumber \\
    & \times\frac{z{\mu}(1-\nu)(h(z)-h(\nu))}{T_\text{F}^2(1-z+z\mu)(z-\nu)},
   \end{align}
   where $h(x)$ is defined in \eqref{df:hx}.
\end{proposition}

\begin{proof}
   See  Appendix \ref{proof:pgf yi}.
\end{proof}

   By employing the L'H\^{o}pital's rule, the average inter-departure time can be readily obtained as
\begin{align}\label{df:ey}
  \mathbb{E}(Y) = &T_\text{F} - \frac{h(\nu)}{T_\text{F}} + \frac{(2\mu + \nu -1)h(\nu)}{\mu(1-\nu)T_\text{F}^2} \nonumber \\
  & + \frac{(1-\mu-\nu)(T_\text{F}^4 + h(\nu))}{\mu(1-\nu)T_\text{F}}.
  \end{align}

Note that the system time $T_{k-1}$ of update $k-1$ can be expressed as $T_{k-1} = W_{k-1}+ S_{k-1}$.
    Since the waiting time $W_{k-1}$ is correlated with inter-arrival time $X_k$,  $T_{k-1}$ would also be correlated with $X_k$.
On the correlation between $X_k$ and  $T_{k-1}$ , we have the following proposition.

\begin{proposition}\label{prop:e_xw}
   The average of $X_kW_{k}$ is given by
\begin{equation}\label{df:e_xw}
    \mathbb{E}[X_kW_{k}] = \frac{{\nu}h'(\nu)}{T_\text{F}^2(1-\nu)},
\end{equation}
where $h(x)$ is defined in \eqref{df:hx} and $h'(x)$ is its derivative.
\end{proposition}

\begin{proof}
  See Appendix \ref{proof:e_xw} .
\end{proof}

\begin{figure*}[htp]   

\hspace{-6 mm}
    \begin{tabular}{cc}
    \subfigure[Transmission probability $p_\text{tx}$ versus $\rho$]
    {
    \begin{minipage}[t]{0.5\textwidth}
    \centering
    {\includegraphics[width = 3.25in] {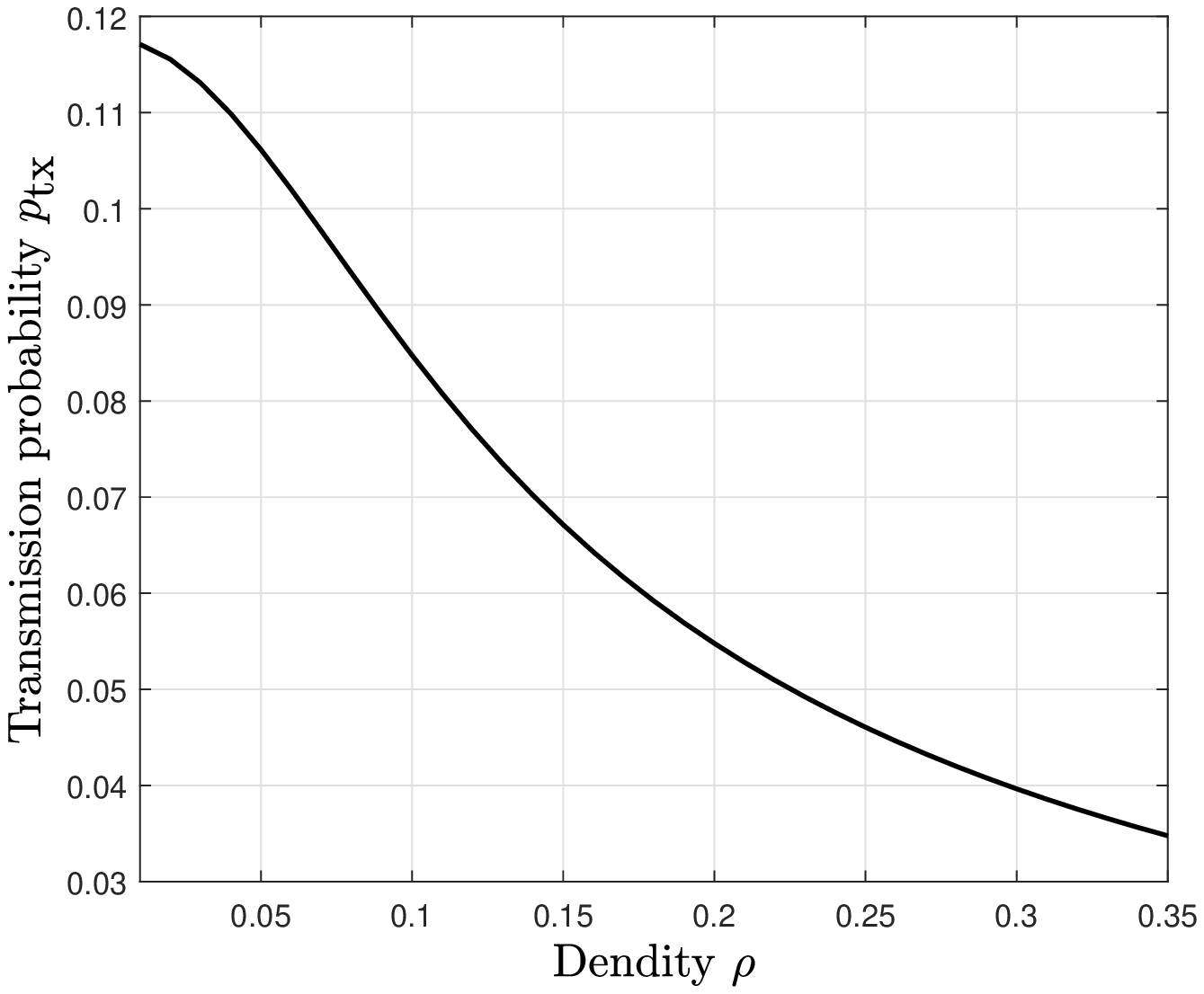} \label{fig:ptx_rho}}
    \end{minipage}
    }

    \subfigure[Collision probability $p_\text{cl}$ versus  $\rho$]
    {
    \begin{minipage}[t]{0.5\textwidth}
    \centering
    {\includegraphics[width = 3.25in] {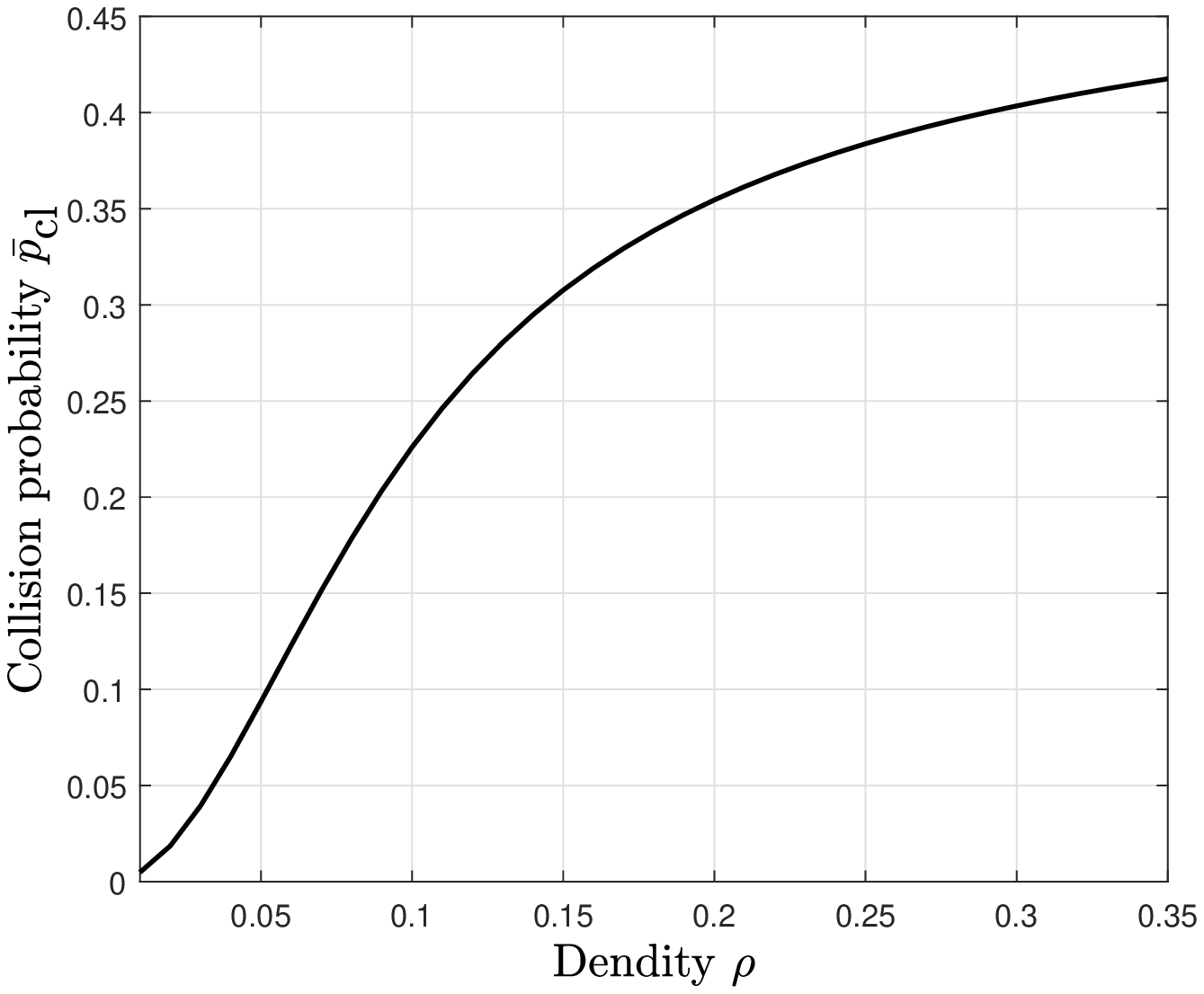} \label{fig:pcl_rho}}
    \end{minipage}
    }
    \end{tabular}
\caption{Average AoI of the power-splitting system. } \label{fig:w_aoi}
\end{figure*}

Based on \textit{Proposition}~\ref{prop:pgf yi} and \textit{Proposition}~\ref{prop:e_xw}, we have the following theorem on the average BAoI of the network.

\begin{theorem} \label{BAoI}
   Note that $\bar{\Delta}_{\text{BC}}$ can be expressed as
\begin{equation}
   \bar{\Delta}_\text{BC} =\frac{\frac{T_\text{F}}{2} + \frac{7T_\text{F}^2-1}{12} + \frac{T_\text{F}}{\mu} + \mathbb{E}(XW) }{\mathbb{E}(Y)} ,
\end{equation}
where $\mathbb{E}(Y)$ and $\mathbb{E}(XW)$ are given, respectively, by \eqref{df:ey} and \eqref{df:e_xw}.
\end{theorem}
\begin{proof}
See Appendix~\ref{proof:BAoI}
\end{proof}

\begin{remark}
In our model, a received update from some other node would replace a node's own generated update of the frame.
        Since the service time of the update is geometrically distributed, the time receiving an update would be uniformly distributed in a frame, i.e., following the same distribution as the generation time of an update.
    Thus, the average BAoI of the first hop broadcasting an update would be equal to that of the following hops.
That is, the average BAoI is additive and our model/results can be readily applied to multi-hop broadcasts.
\end{remark}

\begin{remark}
    Note that the unit of BAoI is second per hop.
        That is, BAoI presents the required time for an update to be spread from a node to its one-hop neighbors.
    Thus, we can refer to its inverse $v=\frac{1}{\bar{\Delta}_\text{BC}}$ as the \textit{velocity of information spreading}, which has a unit of hop/second.
        Since BAoI is the same for each hop of the spreading process of each update, $v$ can be seen as an essential parameter of the network.
\end{remark}

\section{Numerical Results}\label{sec:nuresults}
In this section, we investigate  the average BAoI via numerical results.
     We set the minimum contend window to $w_\text{min}  = 16$ and transmit range to $r =4$ m.
We set the frame length to $T_\text{F}=50$ slots.

In Figs. \ref{fig:ptx_rho} and \ref{fig:pcl_rho}, we present  $p_\text{tx}$ and $p_\text{cl}$  as  functions of density $\rho$.
     To ensure that equation $z=G_X[1-\mu(1-z)]$ has valid solution in$(0,1)$ (see \eqref{eq:alpha_equation}),  the node density should be no larger than $0.35$, i.e., $\rho<0.35$.
First,  we see is that $p_\text{tx}$  is decreasing with  $\rho$ while $\bar{p_\text{cl}}$ is increasing with $\rho$.
     This is because when $\rho$ becomes larger,  each node would have more neighbor nodes and more transmission contention.
Thus, the collision probability would be increased.
    The increment in transmission collisions in turn, would result to smaller transmission probability.

\begin{figure}[!t]
  \centering
  \includegraphics[width = 3.25in]{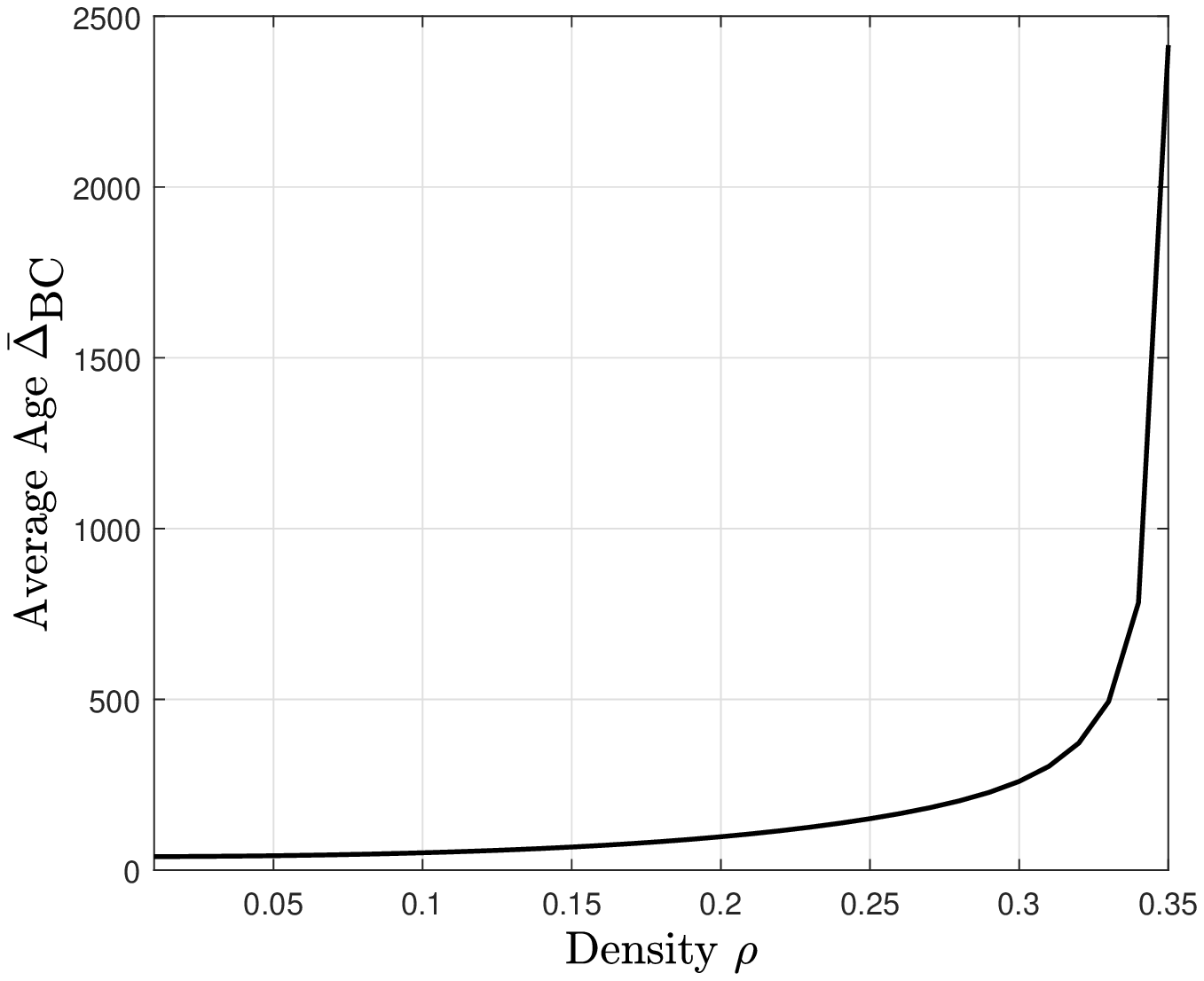}
  \caption{The average BAoI versus $\rho$.}
  \label{fig:baoi_rho}
\end{figure}

\begin{figure}[!h]
  \centering
  \includegraphics[width=3.25in]{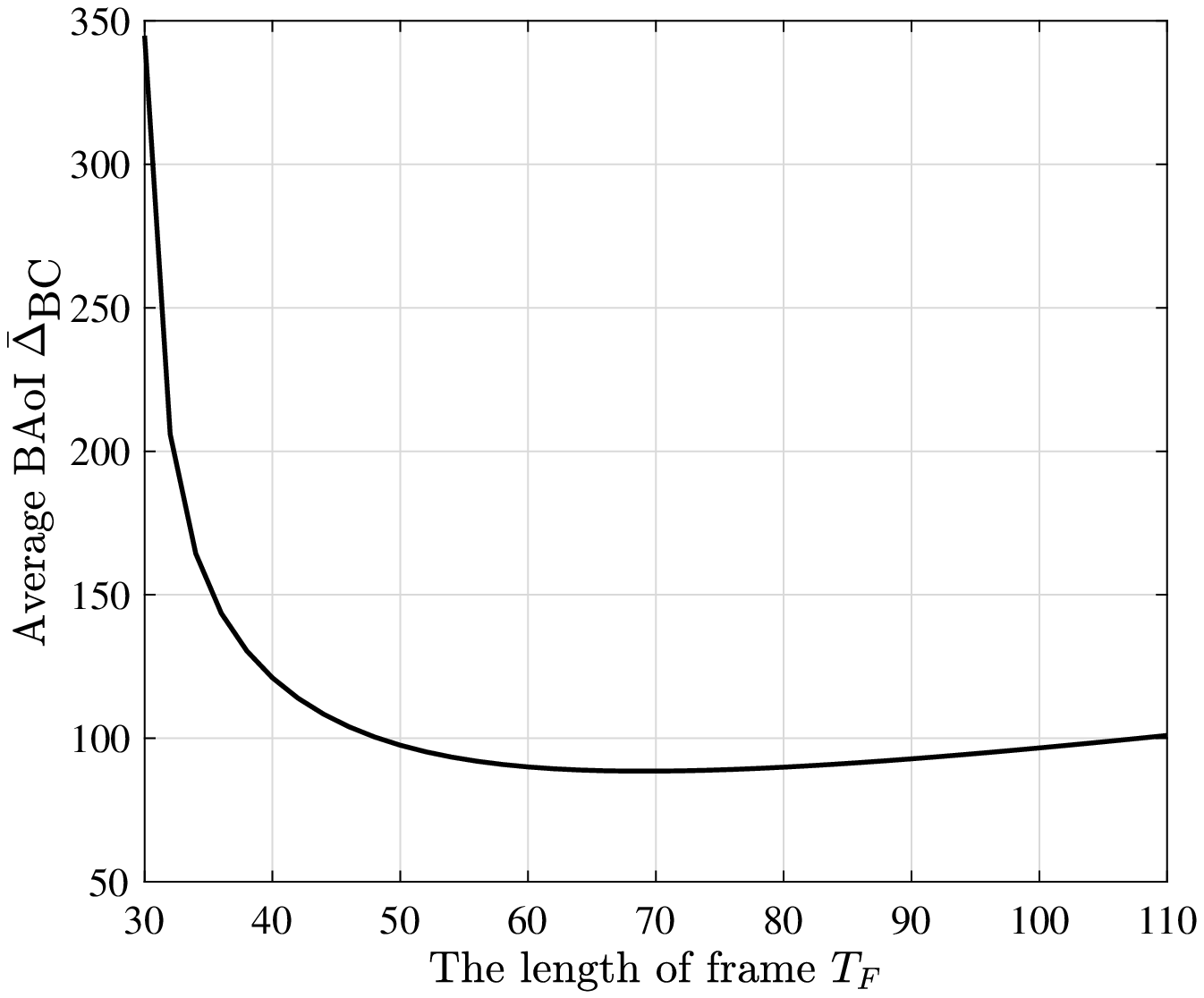}\\
  \caption{The average Broadcast age of information versus $T_\text{F}$.}
  \label{fig:baoi_tf}
\end{figure}

We then investigate how the BAoI changes with node density $\rho$  in Fig. \ref{fig:baoi_rho}.
    We observe that average BAoI is increasing with $\rho$.
In particular, as $\rho$ approaches the maximum density ensuring a reasonable solution to $z=G_X[1-\mu(1-z)]$, the average BAoI would go to infinity.
    This is because when the node density is increased, collisions would occur more frequently and the service time of each update would be larger.

We also how the frame length affects the average BAoI in Fig. \ref{fig:baoi_tf},  where $\rho = 0.2$ is used.
    Since each node would generate one update in each frame, varying the frame length is equivalent to change the traffic density of the network.
As is shown, the average BAoI is decreasing with $T_\text{F}$ first and is increasing with $T_\text{F}$ when it gets large.
    Specifically, the average BAoI is large when $T_\text{F}$ is small, i.e., traffic density is high.
Also, when $T_\text{F}$ gets very large, the updates would be generated very infrequently and the corresponding average BAoI is large.
    To reduce the BAoI of the network, therefore, the frame length should neither be to large nor too small.

\section{Conclusion}\label{sec:conclusion}
In this paper, we have studied the timeliness of the update spreading over a CSMA/CA based wireless network.
    The proposed broadcast age of information measures the age of each successful broadcasting from a transmitter perspective and specifies the speed of the information dissemination process  from each update source to remote areas.
In this paper, we have assumed that each update can be transmitted in one slot.
    For the case each update requires more slots of transmissions, the average BAoI may be obtained based on queueing with server vacations and will be studied in our future work.
Moreover, a node would forward every update it receives in our model, which may result in the broadcast storm problem when traffic rate is high.
    Thus, we shall also consider the performance of probabilistic flooding in  future.

\appendix
In the following proofs, we use the following three auxiliary functions for notational convenience.

\begin{equation}
  f(x)=\sum_{k=1}^{T_\text{F}}x^k k = \frac{x-(1+T_\text{F})x^{T_\text{F}+1}  + T_\text{F}x^{T_\text{F}+2}}{(1-x)^2},\\
\end{equation}
\begin{align}
   g(x)&=\sum_{k=T_\text{F}+1}^{2T_\text{F}-1}x^k (2T_\text{F}-k) \nonumber \\
       &=\frac{(T_\text{F}-1)x^{T_\text{F}+1}-T_\text{F}x^{T_\text{F}+2} + x^{{2T_\text{F}+1}}}{(1-x)^2},
\end{align}
and
\begin{align}
  h'(x) = &\frac{1}{(1-x)^3}\Big( (2T_\text{F}+1)x^{2T_\text{F}} -(2+2T_\text{F})x^{T_\text{F}} +1+x\nonumber \\
  &+(2T_\text{F}-2)x^{T_\text{F}+1} +(1-2T_\text{F})x^{2T_\text{F}+1} \Big)
\end{align}
is the derivative of $h(x)$ (c.f. \eqref{df:hx}) with respect to $x$.

\subsection{Proof of \textit{Proposition} \ref{prop:pgf_xi}}
\begin{proof}\label{proof:pgf_xi}
   The PGF of the interarrival times can be obtained as follows.
\begin{align}
   G_X(z) &=\sum^{2T_\text{F}-1}_{k=1} z^k \Pr \{X_i=k\} \nonumber \\
          &=\sum^{T_\text{F}}_{k=1}z^k\frac{k}{T_\text{F}^2}
          +\sum^{2T_\text{F}-1}_{k=T_\text{F}+1}\frac{2T_\text{F}-k}{T_\text{F}^2} \nonumber\\
          &= \frac{h(z)}{T_\text{F}^2},
\end{align}
where $h(z)$ is defined in \eqref{df:hx}.
\end{proof}

\subsection{Proof of \textit{Proposition} \ref{prop:pgf yi}}
\begin{proof} \label{proof:pgf yi}
     First, the probability that inter-arrival time $X_k$ is smaller than the system time $T_{k-1}$ of previous update can be obtained as follows.
\begin{align}
    \Pr \{X_k < T_{k-1}\} &= \sum^\infty_{i=1}\Pr (X_k = i)\sum_{j=i+1}^{\infty}\Pr ( T_k = j) \nonumber \\
                              &= \sum^\infty_{i=1}\Pr (X_k = i)\nu^i \nonumber \\
                              &= \frac{1}{T_\text{F}^2}\Bigg(\sum_{i=1}^{T_\text{F}}i\nu^i  + \sum_{i = T_\text{F}+1}^{2T_\text{F}-1}(2T_\text{F}-i)\nu^i\Bigg)\nonumber \\
                              &= \frac{1}{T_\text{F}^2}h(\nu).
\end{align}

Second, the MGF of $Y_k$ conditioned on $X_k \geq T_{k-1}$ is given by
\begin{equation}
    \mathbb{E}[z^{Y_k}|X_k > T_{k-1}] =\mathbb{E}[z^{X_k - T_{k-1}}]\mathbb{E}[z^{S_k}],
\end{equation}
where $\mathbb{E}[z^{X_k - T_{k-1}}] $ and $\mathbb{E}[z^{S_k}]$  can be, respectively, derived as follows.
\begin{align}
    &\mathbb{E}[z^{X_k - T_{k-1}}] \nonumber \\
    & = \sum_{i=1}^{2T_\text{F}-1}\Pr \{X_k =i \}\sum_{j=1}^{i}\Pr \{T_{k-1}= j\} \nonumber\\
    &~~~~\cdot\mathbb{E}[z^{i-j}|X_k =i,T_{k-1}=j] \nonumber \\
    & = \frac{1-\nu}{T_\text{F}^2(z-\nu)}(h(z)-h(\nu)),
\end{align}

\begin{align}
    \mathbb{E}[z^{S_k}] &= \sum_{k=1}^{\infty}z^k \Pr \{S_i=k\} \nonumber \\
    & = \sum_{k=1}^{\infty}z^k\mu(1-\mu)^{k-1} \nonumber \\
    & = \frac{z\mu}{1-z+z\mu}.
\end{align}
Since the inter-departure time $Y_k$ can be expressed as
  \begin{equation}  \label{eq:yk}
          Y_k=\left\{
                    \begin{aligned}
                    &S_k,        &if X_k < T_{k-1} , \\
                    &X_k + S_k - T_{k-1} ,  &if X_k \geq T_{k-1},
                    \end{aligned}
    \right.
 \end{equation}
   we have
  \begin{align}
     G_Y(z) =&\Pr \{X_k \geq T_{k-1}\}\mathbb{E}[z^{Y_k}] + \Pr \{X_k < T_{k-1}\}\mathbb{E}[z^{Y_k}] \nonumber \\
     & = \Pr \{X_k \geq T_{k-1}\}\mathbb{E}[z^{X_k -T_{k-1}}]\mathbb{E}[z^{S_k}]\nonumber \\
     & + \Pr \{X_k < T_{k-1}\}\mathbb{E}[z^{S_k}],
  \end{align}
   where the second equation follows \label{eq:yk}and the fact that $S_k$ is independent with $X_k$ and $T_{k-1}$.
   By combining the above results, the proof of the proposition would be completed readily.

\end {proof}

\subsection{Proof of \textit{Proposition} \ref{prop:e_xw}}
\begin{proof} \label{proof:e_xw}
    Since $W_{k} = \max(0,T_{k-1}-X_{k})$ and  $X_k$ is correlated with $W_k$,
     we shall need an auxiliary function $G_(z)$ in calculating $\mathbb{E}[X_kW_k]$ as follows.
    \begin{align}
      G(z) &=  \sum_{i=1}^{2T_\text{F}-1}z^i\Pr (X_{k}=i)\sum_{j=0}^{\infty}\Pr
      \{T_{k-1}>i+j\} \nonumber \\
      &=\sum_{i=1}^{2T_\text{F}-1}z^i\Pr (X_{k}=i)\sum_{j=i}^{\infty}\Pr
      \{T_{k-1}>j\} \nonumber \\
                &=\sum_{i=1}^{2T_\text{F}-1}z^i\Pr (X_{k}=i)\frac{\nu^{i}}{1-\nu} \nonumber \\
                &= \frac{\sum_{i=1}^{T_\text{F}}(z\nu)^ii + \sum_{i=T_\text{F}+1}^{2T_\text{F}-1}(z\nu)^i(2T_\text{F}-i)}{T_\text{F}^2(1-\nu)}
                \nonumber \\
                &=\frac{h(z\nu)}{T_\text{F}^2(1-\nu)},
    \end{align}
    where $h(x)$ is defined in \eqref{df:hx}.
    With $G(z)$, the desired result would be obtained readily.
\begin{align}
   \mathbb{E}(X_kW_k) &= \mathbb{E}(X_{k}\mathbb{E}(\max(0,T_{k-1}-X_{k}))\nonumber\\
                  &=\sum_{i=1}^{2T_\text{F}-1}i\Pr (X_{k}=i)\sum_{j=0}^{\infty}\Pr
      \{max(0,T_{k-1}-i)>j\} \nonumber \\
                  &=\sum_{i=1}^{2T_\text{F}-1}i\Pr (X_{k}=i)\sum_{j=0}^{\infty}\Pr
      \{T_{k-1}>i+j\} \nonumber \\
                  &=\lim_{z=1^{-}}(G(z))' \nonumber \\
                  &=\frac{{\nu}h'(\nu)}{T_\text{F}^2(1-\nu)}.
\end{align}
\end{proof}

\subsection{Proof of  \textit{Theorem} \ref{BAoI}}
\begin{proof} \label{proof:BAoI}
   By expressing the average BAoI as
   \begin{align}
     \bar{\Delta}_\text{BC}&=\lim_{M\to \infty}\frac{K-1}{M}\mathbb{E}[\frac{X}{2}+\frac{X^2}{2}+XT] \nonumber  \\
     &=\frac{\frac{\mathbb{E}[X]}{2}  + \frac{\mathbb{E}[X^2]}{2}  + \mathbb{E}[XS]+\mathbb{E}[XW]}{\mathbb{E}(Y)}  \nonumber \\
     &=\frac{\frac{T_\text{F}}{2} + \frac{7T_\text{F}^2-1}{12} + \frac{T_\text{F}}{\mu} + \mathbb{E}(XW) }{\mathbb{E}(Y)},
   \end{align}
   the proof of \textit{Theorem} \ref{BAoI} is completed readily.
\end{proof}

\small{
\bibliographystyle{IEEEtran}

\begin{thebibliography}{11}

\bibitem{Yates-2012-age}
S. K. Kaul, R. D. Yates, and M. Gruteser, ``Real-time status: How often should one update?" in \textit{Proc. IEEE INFOCOM}, Orlando, FL, USA, Mar. 2012, pp. 2731--2735.

\bibitem{Dong-2018-two-way}
Y. Dong, Z. Chen, and P. Fan, ``Uplink age of information of unilaterally powered two-way data exchanging systems," in \textit{Proc. IEEE Int. Conf. Comput. Commun. Wkshp. (Infocom Wkshp¡¯18)}, Honolulu, HI, US, Apl. 2018, pp. 559--564.

\bibitem{Sun-2016-mlt-sver}
A. M. Bedewy, Y. Sun, and N. B. Shroff, ``Optimizing data freshness, throughput, and delay in multi-server information-update systems," \textit{Proc. IEEE Int. Symp. Inf. Theory (ISIT)}, Barcelona, Spain, Jul. 2016, pp.  2569--2573.


\bibitem{Bian-2000-JSAC-CSMA-th}
G. Bianchi, ``Performance analysis of the IEEE 802.11 distributed coordination function," \textit{IEEE J. Sel. Areas  Commun.,} vol. 18, no. 2, pp. 535--547, Mar. 2000.

\bibitem{Babak-2011-TVT-CSMA-th}
H. Mansouri, M.  Pakravan, and B. Khalaj, ``Analytical modeling and performance analysis of flooding in CSMA-based wireless networks," \textit{IEEE Trans. Veh. Technol.,} vol. 60, no. 2, pp. 664--679, Feb. 2011.

\bibitem{Kar-2005-Infocom-CSMA-fair}
X. Wang and K. Kar, ``Throughput modelling and fairness issues in CSMA/CA based ad-hoc networks," in \textit{Proc. Annu. Joint Conf. IEEE Comput. Commun. Soc.,} Miami, Fl. USA, Mar. 2005, pp. 23--34.

\bibitem{Modiano-2000-TON-broadcst}
A. Sinha, G. Paschos, C. Li, and E. Modiano, ``Throughput-optimal multihop broadcast on directed acyclic wireless networks," \textit{IEEE/ACM Trans. Netw.,} vol. 25, no. 1, pp. 377-391, Jan. 2017.

\bibitem{book_queueing_2008}
N. Tian and X. Xu, \textit{Discrete time queuing theory.,} Beijing, China, Science Press, 2008, pp. 129--134.

\bibitem{G-1996}
G. Bianchi, L. Fratta and M. Oliveri, ``Performance evaluation and enhancement of the CSMA/CA MAC protocol for 802.11 wireless LANs," \textit{Proceedings of PIMRC '96 - 7th International Symposium on Personal, Indoor, and Mobile Communications.,} Taipei, Taiwan, 1996, pp. 392-396 vol.2.

\bibitem{M-2014}
 M. Costa, M. Codreanu, and A. Ephremides, ``Age of information with packet management," in \textit{2014 IEEE International Symposium on Information Theory.,} June 2014, pp. 1583--1587.

\bibitem{Y-2017}
Y. Sun, E. Uysal-Biyikoglu, R. D. Yates, C. E. Koksal, and N. B. Shroff, ``Update or wait: How to keep your data fresh," \textit{IEEE Transactions on Information Theory.,} vol. 63, no. 11, pp. 7492¨C7508, Nov 2017.

\bibitem{X. Ge}
X. Ge, S. Tu, G. Mao, C.-X. Wang and T. Han, ``5G Ultra-Dense Cellular Networks," \textit{ IEEE Wireless Communications.,} vol. 23, no. 1, pp.72-79, Feb. 2016.

\end{thebibliography}

}

\end{document}